\documentclass[prl,twocolumn,superscriptaddress,showpacs,showkeys,preprintnumbers,amsmath,amssymb]{revtex4}

\usepackage{graphicx}
\usepackage{dcolumn}
\usepackage{bm}

\usepackage{hyperref}
\usepackage{enumerate}
\usepackage{verbatim}
\usepackage{amsmath}
\usepackage{latexsym}
\usepackage{revsymb}
\usepackage{yfonts}
\usepackage{amsfonts}
\usepackage{amsmath}
\usepackage{amssymb}
\usepackage{amsthm}
\usepackage{graphicx}
\usepackage{bm}
\usepackage{bbm}
\usepackage{color}
\usepackage{mathrsfs}
\usepackage{epstopdf}

\def\duzomniejsze{<\kern-.7mm<}
\def\duzowieksze{>\kern-.7mm>}

\def\textbf#1{{\bf #1}}
\def\beq{\begin{equation}}
\def\eeq{\end{equation}}
\def\be{\begin{equation}}
\def\ee{\end{equation}}
\def\ben{\begin{eqnarray}}
\def\een{\end{eqnarray}}
\def\beqa{\begin{eqnarray}}
\def\eeqa{\end{eqnarray}}
\def\eea{\end{array}}
\def\bea{\begin{array}}
\newcommand{\bei}{\begin{itemize}}
\newcommand{\eei}{\end{itemize}}
\newcommand{\bee}{\begin{enumerate}}
\newcommand{\eee}{\end{enumerate}}

\newtheorem{theorem}{Theorem}
\newtheorem{proposition}{Proposition}

\newtheorem{lemma}[theorem]{Lemma}

\begin{document}


\title{Free randomness amplification using bipartite chain correlations}

\author{Andrzej Grudka}
\affiliation{Faculty of Physics, Adam Mickiewicz University, 61-614 Pozna\'{n}, Poland}

\author{Karol Horodecki}
\affiliation{National Quantum Information Center of Gda\'{n}sk, 81-824 Sopot, Poland}
\affiliation{Institute of Informatics, University of Gda\'{n}sk, 80-952 Gda\'{n}sk, Poland}

\author{Micha{\l} Horodecki}
\affiliation{National Quantum Information Center of Gda\'{n}sk, 81-824 Sopot, Poland}
\affiliation{Institute of Theoretical Physics and Astrophysics, University of Gda\'{n}sk, 80-952 Gda\'{n}sk, Poland}

\author{Pawe{\l} Horodecki}
\affiliation{National Quantum Information Center of Gda\'{n}sk, 81-824 Sopot, Poland}
\affiliation{Faculty of Applied Physics and Mathematics, Technical University of Gda\'{n}sk, 80-233 Gda\'{n}sk, Poland}

\author{Marcin Paw{\l}owski}
\affiliation{Institute of Theoretical Physics and Astrophysics, University of Gda\'{n}sk, 80-952 Gda\'{n}sk, Poland}
\affiliation{Department of Mathematics, University of Bristol, Bristol BS8 1TW, U.K.}

\author{Ravishankar Ramanathan}
\affiliation{National Quantum Information Center of Gda\'{n}sk, 81-824 Sopot, Poland}




\date{\today}

\begin{abstract}
A direct analysis of the protocol of randomness amplification using Bell inequality violation is performed in terms of the convex combination of no-signaling boxes required to simulate quantum violation of the inequality. The probability distributions of bits generated by a Santha-Vazirani source are shown to be mixtures of permutations of Bernoulli distributions with parameter defined by the source. An intuitive proof is provided for the range of partial randomness from which perfect randomness can be extracted using quantum correlations violating the chain inequalities. Exact values are derived in the asymptotic limit of a large number of measurement settings.
\end{abstract}

\maketitle

{\it Introduction.} The question whether all processes in Nature are predetermined or if there are fundamentally unpredictable events is a most fundamental one. While it seems impossible to rule out complete determinism at all levels, the philosophical and practical implications such as in gambling and cryptographic scenarios have made it a question worthy of thorough investigation. In this regard, exciting new results have been obtained in \cite{Renner, Acin, Pawlowski} that the correlations in quantum systems can be used to amplify randomness. In particular, it has been shown that the presence of a small amount of unpredictability can be used to infer the presence of truly random events. 

Formally, the information-theoretic task is called randomness amplification, where the goal is to use an input source of partially random bits to produce a perfect random bit. The source of randomness is taken to be the Santha-Vazirani source \cite{SV} which is defined as follows. A source is called a Santha-Vazirani (SV) source if for any random variable $X = (X_1, X_2, \dots, X_n)$ produced by this source and for any $0 \leq i < n$ and $x_i = \{0, 1\}$, there holds 
\be
\frac{1}{2} - \epsilon \leq P(X_{i+1} = x_{i+1} | X_i = x_i, \dots, X_1 = x_1) \leq \frac{1}{2} + \epsilon. 
\label{SVcond}
\ee 
The model can be interpreted as each bit being obtained by the flip of a biased coin, the bias being fixed by an adversary who has knowledge of the history of the process. As such, the conditioning variables can be any set of pre-existing variables $W$ that could be a possible cause of the succeeding bit $X_{i+1}$. Each bit produced by the source is $\epsilon$-free in the sense that the probability distribution is $\epsilon$ away in variational distance from the uniform distribution. The goal of randomness amplification is to produce perfect random bits, i.e., those with $\epsilon_{new} = 0$. Note that randomness amplification differs from the task of (device-independent) randomness expansion, where it is assumed that an input seed of perfect random bits is available and the goal is to expand this given bit string into a larger sequence of random bits. Quantum non-locality has also found application in this latter task \cite{Pironio, Colbeck, Acin2} as well as in device-independent cryptographic scenarios \cite{BHK, Hanggi}.

In \cite{SV}, it was shown that the randomness produced by a single SV source cannot be amplified by classical means, by any deterministic function. The idea behind randomness amplification using quantum correlations in \cite{Renner, Acin} is then to use the SV source to choose the measurement settings of a set of spatially separated observers in a Bell test and to obtain random bits from some function of the measurement outcomes. In \cite{Renner}, the bipartite scenario of chained Bell inequalities \cite{BC} was shown to be useful in obtaining perfectly random bits as measurement outcomes for a limited range of $\epsilon$ values ($\epsilon < \frac{(\sqrt{2}-1)^2}{2}$ assuming correctness of quantum theory). In \cite{Acin}, a more complicated five-party scenario using Mermin inequalities \cite{Mermin} was considered and shown to generate perfect random bits for any initial value of $\epsilon < \frac{1}{2}$. The validity of the no-signaling principle is vital in both protocols, in fact no-signaling was shown to be necessary for perfect randomness to occur in any theory. 

A fundamental understanding of the probability distributions of bits generated by the source of partial randomness is necessary to study how and when tasks such as randomness amplification can be performed given different strengths of the adversary. In this paper, we investigate the structure of the SV source showing that the extremal points of the set of probability distributions from such a source are permutations of Bernoulli distributions. Indeed, this fact has already found application in the task of randomness amplification given adversaries limited to quantum resources \cite{Pawlowski}. Moreover, in the search for simpler (possibly bipartite) protocols for generation of perfect random bits from any initial value of $\epsilon$ for any no-signaling adversary, it becomes vital to derive intuitive methods that apply to arbitrary scenarios as well as to understand the limits of applicability of currently known protocols \cite{Renner}. We address both these issues, providing an analysis of the protocol of randomness amplification in terms of the randomness present in the no-signaling boxes that appear in convex decompositions of the quantum box of probabilities. This is then used for a direct derivation of the known range of $\epsilon$ values from which perfect randomness can be generated using the bipartite correlations violating the chained Bell inequalities as well as to extend the result to asymptotically exact values. 

{\it Structure of Santha-Vazirani sources}. The protocol for randomness amplification from SV sources using non-local quantum correlations involves the use of the source to choose the measurement settings in the Bell expression. For example, in the bipartite scenario of the chained Bell inequalities, a string of bits from the source is used to generate the measurement settings $\bf{x}$ and $\bf{y}$ of the two parties. Our aim in this section is to characterize the joint probability distributions $P(\bf{x},\bf{y})$ which can arise from the source, i.e., those that satisfy the SV source conditions (\ref{SVcond}). We investigate the structure of the SV sources and prove that the distributions obeying (\ref{SVcond}) are mixtures of permuted Bernoulli distributions. Formally, we state the following proposition (proof provided in the Supplementary Material) which will be used in the analysis of the randomness amplification protocol in subsequent sections.

{\bf Proposition 1.} Extremal points of the set of probability distributions from a Santha-Vazirani source with parameter $\epsilon$ are permutations of Bernoulli distributions with parameter $p=p_+$, 
where $p_+=\frac12 +\epsilon$. 

{\bf Remark.}  Not all permutations are allowed. 

{\it Randomness amplification from non-local quantum correlations.}  Consider the scenario where the bits generated by the SV source (that are partially free with respect to any set of space-time variables held by an adversary Eve) are used to choose the measurement settings in a Bell test by a set of $N$ spatially separated observers. Upon violation of the inequality, the parties process the measurement outcomes in order to obtain a perfect random bit. The general $N$ party Bell inequality for randomly chosen measurement settings can be written as
\be
\beta = \sum_{\vec{a}, \vec{x}} \alpha(\vec{a}, \vec{x}) P(\vec{a} | \vec{x}) \leq \beta_{L}.
\ee
Here, $\vec{x}$ denotes a set $\{x_1, \dots, x_N\}$ of measurement settings chosen by the $N$ parties, $\vec{a}$ = $\{a_1, \dots, a_N\}$ denotes the respective measurement outcomes, $\alpha(\vec{a}, \vec{x})$ are a set of coefficients and $P(\vec{a} | \vec{x})$ denotes the conditional probability of outcomes $\vec{a}$ given settings $\vec{x}$. The bound $\beta_{L}$ denotes the optimal value of the Bell parameter attainable within local hidden variable (LHV) theories. A quantum state under suitable measurement settings then generates the box of probabilities $B_{Q}$ that leads to optimal violation of the inequality $\beta_{Q}$. In the scenario where the measurement settings are not chosen freely but using an SV source with parameter $\epsilon$, one obtains a new LHV optimal value as a function of $\epsilon$ denoted by $\beta_{L}(\epsilon)$. The adversary Eve may attempt to simulate the value $\beta_{Q}$ using a convex combination of no-signaling boxes $B_{NS}^{(i)}$ which produce values $\beta_{NS}^{(i)}$, i.e., $\beta_{Q} = \sum_{i} p_{i} \beta_{NS}^{(i)}$ with $\sum_{i} p_{i} = 1$. The process of randomness amplification is then transparently based on the randomness present in the boxes $B_{NS}^{(i)}$. If some function of the measurement outcomes (in particular, simply one of the measurement outcomes of a single party \cite{Renner}) is random for all boxes $B_{NS}^{(i)}$ appearing in any convex decomposition, then the parties may use this as the output free random bit completely uncorrelated from Eve. It is immediately seen that to perform free randomness amplification ($\epsilon_{new} = 0$) from any initial value of $\epsilon < \frac{1}{2}$, one requires that the maximum no-signaling violation of the Bell inequality be achievable within quantum theory; if not, Eve may choose a finite fraction of deterministic boxes in the simulation. In general, for any given $\beta_{Q}$ one may write
\be 
\beta_{Q} = (1- \delta) \beta_{NS}^{(r)} + \delta \beta_{NS}^{(n r)},
\ee
where $\beta_{NS}^{(n r)}$ is the optimal violation of the inequality by boxes that do not give randomness and $\delta$ is the maximum fraction of such boxes that an adversary may use to successfully simulate $\beta_{Q}$. One may recast the above in terms of the probability of failure to win a game defined by $\beta$, where the no-signaling boxes with randomness succeed with probability $1$. Then, using  $\eta_{q} = (1- \beta_{Q})$ to denote the minimum probability of failure within quantum theory to win the game and similarly $\eta_{sv}$ to denote the minimum probability of failure to win the game using no-signaling boxes that do not yield randomness (while choosing the settings using the SV source), we obtain 
\be 
\delta \leq \frac{\eta_{q}}{\eta_{sv}}.
\ee
When $\delta = 0$, one obtains the randomness present in the no-signaling boxes.



{\it Randomness amplification using chain inequalities.} In \cite{Renner}, the amplification of randomness using chained Bell inequalities \cite{BC} was investigated. Two results were obtained, the first under the  assumption of correctness of quantum theory, i.e., that the observed distribution of measurement outcomes is as given by the theory, and the second without this restriction. In the former scenario, it was shown that for given $\epsilon < \frac{(\sqrt{2} - 1)^2}{2}$, there exists a protocol that uses $\epsilon$-free bits with respect to any set of space-time variables $W$ to obtain $\epsilon'$-free bits with respect to $W$ for any $0 < \epsilon' \leq \epsilon$. In particular, the quantum correlations between the measurement outcomes for the chain inequality were used to show that the output bit of one party's measurement is arbitrarily close to being uniform and uncorrelated with $W$. Following the general considerations of the previous section, we can now formulate an intuitive and simpler derivation of this result. 

The chained Bell inequality considers the bipartite scenario of two spatially separated parties Alice and Bob who each choose from a set of $N$ measurement settings: $x \in \{0, \dots, N-1\}_{A}$ for Alice and $y \in \{0, \dots, N-1 \}_{B}$ for Bob. Each measurement results in a binary outcome $a \in \{0,1\}$ for Alice and $b \in \{0,1\}$ for Bob. The chained Bell inequality is then written as
\be
\sum_{x = y || x = y +1} P(a \oplus b = 1 |x, y) + P(a \oplus b = 0 |0, N-1) \geq 1,
\label{chain}
\ee
where $\oplus$ denotes addition modulo $2$. Notice that out of the $N^2$ possible measurement pairs, only the $2N$ neighboring pairs where $x = y$ or $x = y + 1$ (sum modulo $N$) forming a chain are considered in the inequality and the LHV bound is obtained from the fact that perfect correlations in the outcomes for the $2N-1$ pairs in the sum automatically implies perfect correlation for the pair $\{0, N-1\}$. Quantum mechanics violates this inequality obtaining a value of $2 N \sin^2(\frac{\pi}{4N})$ which for large $N$ tends to the algebraic limit of $0$. This optimal quantum value is obtained by measuring on the maximally entangled state $|\phi_{+} \rangle = \frac{1}{\sqrt{2}} (|00\rangle + |11\rangle)$ with the measurement settings defined by the bases $\{|0_k \rangle , |1_k\rangle\}$ (for $k = x, y$). Here $|0_k\rangle = \cos{\frac{\phi_k}{2}} |0\rangle + \sin{\frac{\phi_k}{2}} |1 \rangle $, $|1_k\rangle$ = $ \sin{\frac{\phi_k}{2}} |0\rangle - \cos{\frac{\phi_k}{2}} |1 \rangle$ with the angles $\phi_k = \frac{\pi k}{2 N}$. The set of no-signaling boxes for this scenario was studied in \cite{Masanes}, a no-signaling box with precisely the PR-box structure of perfect correlations for the $2N-1$ neighboring pairs in the sum and perfect anti-correlations for the remaining pair exists which in addition to incorporating perfect randomness attains the optimal no-signaling value of $0$. A crucial observation is that if one pair of measurement settings is known to not occur, classical theories can simulate the optimal no-signaling violation of the inequality. 

Ideally the measurement settings would be chosen freely, however in this scenario they are chosen by Alice and Bob each using $r \mathrel{\mathop:}= \log_2 N$ bits from an SV source with non-zero $\epsilon$. The optimal classical strategy by an adversary is to choose the term that equals $1$ in the Bell expression corresponding to the pair of measurements that the SV source provides with minimum probability. One therefore considers the inequality
\ben
&&\sum_{x = y || x = y +1} P(x,y|w) P(a \oplus b = 1 |x, y) + \nonumber \\
&&P(0, N-1|w) P(a \oplus b = 0 |0, N-1) \geq p_{\text{min}}
\een
for each $w$ in the set of space-time variables with which the imperfectly free SV source may be correlated (and which may be thought of as held by the adversary Eve). The bound $p_{\text{min}} = \min_{x,y} P(x,y|w)$ is the minimum probability of a pair of measurement settings chosen by Alice and Bob, ideally $p_{\text{min}}^{ideal} = \frac{1}{2N}$ (for $\epsilon = 0$). 
As in the previous section, Eve tries to simulate the value $\beta_{Q} = \sin^2(\frac{\pi}{4N})$ using no-signaling boxes with randomness which produce $\beta_{NS}^{(r)}$ and those which do not incorporate randomness and give $\beta_{NS}^{(nr)}$. Crucially, for the chained inequalities all the vertices of the no-signaling polytope have been characterized in \cite{Masanes} and it is found that only those boxes with perfect randomness in the outcomes ($\epsilon_{new} = 0$) violate the chain inequality giving $\beta_{NS}^{(r)} = 0$. All other no-signaling boxes do not violate the inequality and produce $\beta_{NS}^{(nr)} \geq p_{\text{min}}$. Therefore, the optimal violation of the inequality that Eve can achieve using any fraction $\delta$ of non-random no-signaling boxes (and fraction $(1-\delta)$ of random ones) is given by $\beta_{sv} = \delta p_{\text{min}}$.



The measurement settings are chosen using $2r$ uses of the imperfect SV source by Alice and Bob, say the first $r$ bits give in binary the setting $x$ for Alice and the next $r$ bits give Bob's setting $y$. The minimum probability of occurrence for any measurement pair among the $N^2$ pairs is then $p_{-}^{2r} = (\frac{1}{2} - \epsilon)^{2r}$. From the set of obtained measurement settings, only those corresponding to the $2N$ neighboring pairs in the chain inequality are retained while the rest are discarded. Therefore, the minimum probability in the sequence of $2N$ pairs, $p_{\text{min}}$ is given by
\be  
p_{\text{min}} = \frac{p_{-}^{2r}}{p_-^{2r} + ||P(\bf{x}, \bf{y})||_{2N-1}},
\ee
where $||P(\bf{x}, \bf{y})||_{2N-1}$ is the $(2N-1)^{th}$ Ky Fan norm of the probability distribution $P(\bf{x}, \bf{y})$ generated by the source, i.e., the sum of the $2N-1$ largest probabilities.
The denominator of the above expression can be bounded from above by $2N p_{+}^{2r}$ where $p_{+} = (\frac{1}{2} + \epsilon)$, since $p_{+}^{2r}$ is the largest probability of occurrence of a bit string of length $2r$ generated by the source. We therefore obtain that the value of the Bell expression simulated by Eve is given by 
\be 
\beta_{sv} = \delta p_{\text{min}} \geq \delta \frac{p_{-}^{2r}}{2^{r+1} p_{+}^{2r}}.
\ee
For consistency with the value obtained in quantum theory, we have $\beta_{sv} \leq \beta_{Q}$, i.e.,
 $\delta \frac{p_{-}^{2r}}{2^{r+1} p_{+}^{2r}} \leq \sin^{2} \left(\frac{\pi}{2^{r+2}}\right)$.
The fraction of non-random boxes $\delta$ approaches $0$ (and perfect randomness is obtained) as we increase the number of measurement settings $N (= 2^{r})$ provided
\be  
\lim_{r \rightarrow \infty} \frac{\pi^2}{8} \frac{p_{+}^{2r}}{2^r p_{-}^{2r}} = 0,
\ee
giving $\frac{(\frac{1}{2} + \epsilon)^2}{2 (\frac{1}{2} - \epsilon)^2} < 1$, thus recovering $\epsilon < \frac{(\sqrt{2} - 1)^2}{2} \approx 0.086$. 

{\it Asymptotically exact bounds on randomness.} Here, we show an improved estimate on the minimum probability $p_{\text{min}}$ for obtaining a pair of measurement settings from the SV source which gives exact values for the range of allowed $\epsilon$ in the asymptotic limit of large $N$. Among the joint probability distributions that satisfy the $SV$ conditions are the {\it extremal} ones which as we have seen are (certain) permutations of the Bernoulli distribution. Our goal is to find the $(2N-1)^{th}$ Ky Fan norm of the Bernoulli distribution, which being the first $2N-1$ maximal probabilities is permutation invariant and therefore the same for all extremal distributions.




The $2N-1$ Ky Fan norm of the Bernoulli distribution $B$ satisfying SV conditions \eqref{SVcond} is:
\be
||B||_{2^{r+1} -1}= \sum_{i=0}^{m} {2r \choose i} p_+^{2r-i}p_-^{i}
\ee 
where $m$ is chosen to obtain the $2^{r+1} - 1$ largest probabilities.
The task of finding $m$ can be reformulated using
\be
m \leq \min_{c} \{c r: \sum_{i=0}^{cr} {2r \choose i} \geq 2^{r+1} -1 \}
\ee 
to finding the minimum $c$ satisfying the inequality above.  

We now state the following Lemma (with proof in the Supplementary Material) which provides bounds on $||B||_{2^{r+1}-1}$, and leads to the asymptotically exact range of $\epsilon$ from which perfect randomness may be extracted. 

{\bf Lemma 2.} The Ky Fan norm of the Bernoulli distribution $||B||_{2^{r+1}-1}$ with parameter $p_{-} = (1/2 - \epsilon)$ for large $r$ obeys 
\be 
{2r \choose cr} p_{-}^{cr} p_{+}^{(2-c)r} < ||B||_{2^{r+1} - 1} < k {2r \choose cr} p_{-}^{cr} p_{+}^{(2-c)r},
\label{norm-bound}
\ee
where $c$ is the solution to $2^{2r H(c/2)} = 2^r$ ($c \approx 0.22$) and $k = \frac{(2-c) (1-2\epsilon)}{2(1-c -2\epsilon)}$. 

The above Lemma is now used to find when the upper bound on $\delta$ approaches zero, i.e., when
\be 
\lim_{r\rightarrow \infty} \frac{\pi^2}{16} \frac{||B||_{2^{r+1}-1}}{2^{2r} p_{-}^{2r}} = 0
\ee
The bounds in \eqref{norm-bound} imply that the limit is defined by the behavior for large $r$ of
\be 
\frac{ {2r \choose cr} p_{-}^{cr} p_{+}^{(2-c)r}}{2^{2r} p_{-}^{2r} }\approx 2^{2rH(c/2)} \frac{p_{+}^{(2-c)r}}{2^{2r} p_{-}^{(2-c)r}}.
\ee
For the limit to be $0$, we need $(1/2+\epsilon)^{2-c} < 2 (1/2 - \epsilon)^{2-c}$ for $c \approx 0.22$ 
giving $\epsilon < \frac{2^{1/(2-c)} - 1}{2 (2^{1/(2-c)} + 1)} \approx 0.0961$ which is the asymptotically exact maximal value for $\epsilon$ due to the upper and lower bounds derived on the Ky Fan norm. In fact, for $\epsilon$ larger than this critical value, free randomness cannot be obtained in the protocol, i.e., the newly generated randomness value $\epsilon_{new} > \epsilon$ for this range as shown below. 

Note that the amount of randomness $\epsilon_{new}$ obtained in the protocol is given by $(1/2) + \epsilon_{new} = (1-\delta) \times (1/2) + \delta \times 1$, i.e., $\epsilon_{new} = \delta/2$, so that for finite $r$,
\be 
\epsilon_{new} = \frac{\sin^2(\frac{\pi}{2^{r+2}})}{2 \frac{p_-^{2r}}{p_-^{2r} + ||B||_{2^{r+1}-1}}}.
\ee
Using the inequality, 
$\sum_{i=0}^{\lfloor c r \rfloor} {2r \choose i} \leq 2^{H(c/2) 2r}$ for finite $r$,
we see that $c \geq 0.22$ in the definition of $||B||_{2^{r+1}-1}$ for any finite $r$. We may therefore lower bound $\epsilon_{new}$ as
\be  
\epsilon_{new} \geq  \frac{\sin^2(\frac{\pi}{2^{r+2}}) (p_-^{2r} + (2^{r+1} -1) p_-^{0.22r} p_+^{1.78r})}{2 p_-^{2r}}.
\ee
Using $\sin{x} \geq x - x^3/6$, we see after some algebraic manipulation that for $\epsilon > 0.0961$, the lower bound above exceeds $\epsilon$ meaning no amplification is possible. 

{\it Conclusions.} Randomness amplification from quantum correlations violating Bell inequalities is shown to be directly related to the fraction of no-signaling boxes incorporating randomness which appear in any possible convex combination of boxes simulating the violation. The bipartite scenario of the chained Bell inequalities (which includes the commonly considered CHSH inequality) is arguably the most studied as well as being experimentally friendlier than multi-party scenarios. An intuitive and simple derivation is provided for the range of partial randomness from which perfect randomness can be generated using quantum correlations violating these inequalities. In addition, asymptotically exact bounds obtained on the minimum probability of a pair of measurement settings from an SV source enable us to identify the most imperfect source from which perfect randomness can still be generated using these correlations. The characterization of the probability distributions obeying the SV source conditions performed here is of independent interest in general scenarios as well \cite{Pawlowski}.

{\it Acknowledgments.}
K. H. and M. H. thank Renato Renner for introducing them to the problem of amplification of randomness. 
The paper is supported by ERC AdG grant QOLAPS and  by
Foundation for Polish Science TEAM project cofinanced
by the EU European Regional Development Fund. Part of this work was done in National Quantum Information Center of Gda\'{n}sk. K. H. would also like to acknowledge discussion with Justyna {\L}odyga.

\bibliographystyle{apsrev}


{\bf Supplementary Material.}
Here, we present the formal proof of the proposition and lemma stated in the text.

\begin{proposition}
Extremal points of the set of probability distributions from Santha-Vazirani  source
are permutations of Bernoulli distributions with parameter $p=p_+$, 
with $p_+=\frac12 +\epsilon$. 
\label{prop:SV-iid}
\end{proposition}

{\bf Remark.}  Not all permutations are allowed.

To prove the proposition we will need the following lemma.
\begin{lemma} 
Consider two alphabets $X$ and $Y$, with $|X|=K,|Y|=M$.
Consider some convex sets $S_X$ and $S_Y$ of the probability distributions 
over the spaces $X$ and $Y$ respectively. Consider an arbitrary joint probability distribution $p(x,y)$.
Let $p(y|x)$ be the corresponding conditional probability distribution and $p(x)$ the marginal one. 
Suppose now that  for any fixed $x$, the distribution $\{p(y|x)\}_y$  belongs to $S_Y$,
and the distribution $\{p(x)\}_x$ belongs to $S_X$. Then we can write $p(x,y)$ 
as a mixture of probability distributions of the form
\be
\tilde p(x,y) = p^{(x)}(y) r(x)
\ee
where distribution $p^{(x)}$ is extremal in the set $S_Y$ and distribution $r$ is extremal in set $S_X$.
\label{lem:cond_ext}
\end{lemma}

\begin{proof}

Clearly, it is enough to prove that $p(x,y)$ can be written as mixture of distributions 
\be
p'(x,y) =p(x)  p^{(x)}(y) 
\label{eq:needed_form}
\ee
where $p^{(x)}$ is extremal in $S_Y$.
Indeed, then we can decompose $p(x)$ into extremal points in $S_X$, and reach the form \eqref{lem:cond_ext}.

Let $p^{(i)}$ run over extremal elements of $S_Y$. We define the following distributions 
\be
p_{i_1,\ldots i_K} (x,y) =  p(x) p^{(i_x)}(y)
\ee
Clearly they are of the required form \eqref{eq:needed_form}.
We will now show that a suitable mixture of such distributions 
gives $p(x,y)$. To see this, note that since for each $x$, the distribution $p(y|x)$ 
belongs to $S_Y$, we can write it as a mixture of $p^{(i)}$'s 
\be
p(y|x)= \sum_i \lambda^{(x)}_i p^{(i)}(y) 
\ee
where 
\be
\sum_i \lambda^{(x)}_i=1,
\label{eq:lambda}
\ee
for each $x$. 
We will now show that 
\be
p(x,y)= \sum_{i_1,\ldots, i_N} \lambda^{(1)}_{i_1} \cdot \ldots\cdot \lambda^{(N)} _{i_N} 
p_{i_1,\ldots i_K} (x,y)
\label{eq:core}
\ee
which is what we need to prove, as by \eqref{eq:lambda}
we have 
\be
\sum_{i_1,\ldots, i_N} \lambda^{(1)}_{i_1} \cdot \ldots\cdot \lambda^{(N)} _{i_N} =1
\ee
and $p_{i_1,\ldots i_K} (x,y)$ are of the required form  \eqref{eq:needed_form}.
To prove the equality \eqref{eq:core}, we write
\ben
&&\sum_{i_1,\ldots, i_N} \lambda^{(1)}_{i_1} \cdot \ldots \cdot \lambda^{(N)} _{i_N} 
p_{i_1,\ldots i_K} (x,y) =  \nonumber \\
&&\sum_{i_1,\ldots, i_N} \lambda^{(1)}_{i_1} \cdot \ldots \cdot \lambda^{(N)} _{i_N} p(x) p^{(i_x)}(y)  =\nonumber \\
&&  \sum_{i_x} \lambda^{(x)}_{i_x} p(x) p^{(i_x)}(y)  = p(x) p(y|x)  = p(x,y).
\een
The last but one equality we obtain from the fact that only for index $i_x$ the summand 
is nontrivial, for other indices the summands are just $\lambda$'s which sum up to 1.
\end{proof}

Now we are in position to prove the proposition \ref{prop:SV-iid}. 
\begin{proof} [Proof of proposition  \ref{prop:SV-iid}] 
To prove the proposition,  we will apply the lemma \ref{lem:cond_ext} iteratively. The set $X$ will be the set of $n$ bits,
while the set $Y$ will correspond to a single bit. $S_Y$ then has two extremal points $(p_+, p_-)$ and $(p_-, p_+)$.
Let us first illustrate the lemma for the case of $X$ also being a single bit. 
Then simply, 
\be
\{p(x,y)\} = \Bigl(p(0) p(0|0) , p(0)  p(1|0), p(1)p(0|1), p(1)p(1|1) \Bigr)
\ee
Now, for $x=0$, we have decomposition 
\be
p(0|0)= \alpha_0 p_+  + (1-\alpha_0) p_- , \quad p(1|0) = \alpha_0 p_-  + (1-\alpha_0) p_+.
\ee
For $x=1$ we have some other decomposition 
\be
p(0|1)= \alpha_1 p_+  + (1-\alpha_1) p_- , \quad p(1|1)= \alpha_1 p_-  + (1-\alpha_1) p_+
\ee
To catch up with notation of the lemma, we have $\alpha_0=\lambda^0_1, 1-\alpha_0=\lambda^0_2$, 
and  $\alpha_0=\lambda^{(0)}_1, 1-\alpha_0=\lambda^{(0)}_2$, and $p^{(1)}=(p_+,p_-),p^{(2)}=(p_-,p_+)$
are extemal points from $S_Y$.
We can directly check that 
\ben
&&\bigl( p(0) p(0|0), p(0) p(1|0), p(1) p(0|1), p(1) p(1|1) \bigr)=\nonumber\\
&& \alpha_0 \alpha_1 \bigl( p(0)p_+, p(0) p_-,p(1)p_+ ,p(1)p_-  \bigr)+\nonumber\\
&& \alpha_0 (1-\alpha_1) \bigl(p(0)p_+ , p(0)p_- ,p(1)p_- , p(1)p_+ \bigr)+\nonumber\\
&& (1-\alpha_0) \alpha_1 \bigl(p(0)p_- , p(0)p_+ ,p(1)p_+ ,p(1)p_-  \bigr)+\nonumber\\
&& (1-\alpha_0) (1-\alpha_1) \bigl(p(0)p_- , p(0)p_+ , p(1)p_-, p(1)p_+ \bigr) \nonumber \\
\een
Thus we have shown explicitly, decomposition of $p(x,y)$ into distrbutions of the form  \eqref{eq:needed_form}. 
Now we further decompose 
the distribution $(p(0),p(1))$  into extremal points of $S_X$ which are in  this case the same 
as those of $S_Y$:  $(p_+,p_-)$ and $(p_-,p_+)$. 
Therefore  $\{p(x,y) \}$  is mixture of the eight probability distributions 

\ben
\bigl(p_+p_+, p_+p_-, p_-p_+, p_-p_- \bigr),\quad
\bigl(p_+p_+, p_+p_-, p_-p_-, p_-p_+ \bigr),\nonumber \\
\bigl(p_+p_-, p_+p_+, p_-p_+, p_-p_- \bigr),\quad
\bigl(p_+p_-, p_+p_+, p_-p_-, p_-p_+ \bigr),\nonumber \\
\bigl(p_-p_+, p_-p_-, p_+p_+, p_+p_- \bigr),\quad
\bigl(p_-p_+, p_-p_-, p_+p_-, p_+p_+ \bigr),\nonumber\\
\bigl(p_-p_-, p_-p_+, p_+p_+, p_+p_- \bigr),\quad
\bigl(p_-p_-, p_-p_+, p_+p_-, p_+p_+ \bigr),\nonumber\\
\een
where the ordering is as follows: 
\be
\bigl(p(0,0), p(0,1),p(1,0),p(1,1)\bigr).
\ee
Note that the first distribution is precisely the Bernoulli distribution, 
with probability of $0$ in single trial being $p=p_+$. This distribution is memoryless. 
The other distributions are not memoryless, but are related to the Bernoulli distribution by permutation of probabilities (not bits). Note that only 8 out of 24 permutations appear.

For $n$ bits, the lemma implies that the extremal probability distributions 
are created from product of the extremal distributions for $n-1$ bits, 
as follows. For given extremal distribution $\bigl(r(1), \ldots, r(K)\bigr)$ with $K=2^{n-1}$, we construct the following 
extremal point:
\be
\bigl(r(1) p_+, r(1) p_-, r(2) p_+, r(2) p_-, \ldots, r(K) p_+, r(K) p_-\bigr).
\label{eq:extr-bern}
\ee
The other extremal points can be generated from it by changing the order of $p_+$ and $p_-$  
for each $x=1,\ldots, K$ independently. This implies, that all the extremal points are 
permutations of the  above one.
Now, by induction we assume that the distribution $\bigl(r(1), \ldots, r(K)\bigr)$ over $n-1$ bits is a permutation 
of Bernoulli distribution over $n$ bits with parameter $p=p_+$. Thus, there is permutation $\sigma$ 
that reorders it, so that it becomes Bernoulli. We can apply this permutation to reorder 
pairs $(r(i) p_+, r_i p_-)$ in the distribution \eqref{eq:extr-bern}. The resulting distribution 
is Bernoulli for $n$ bits with parameter $p_+$. Thus \eqref{eq:extr-bern} is permutation of Bernoulli, and hence 
all other extremal points are, too, since they are its permutations. 
Note that not all permutations are allowed, because the above construction has the structure of a tree. 
\end{proof}

\begin{lemma} The Ky Fan norm of the Bernoulli distribution $||B||_{2^{r+1}-1}$ with parameter $p_{-} = (1/2 - \epsilon)$ is bounded for large $r$ by 
\be 
{2r \choose cr} p_{-}^{cr} p_{+}^{(2-c)r} < ||B||_{2^{r+1} - 1} < k {2r \choose cr} p_{-}^{cr} p_{+}^{(2-c)r},
\ee
where $c$ is the solution to $2^{2r H(c/2)} = 2^r$ ($H(x)$ denotes binary entropy giving $c \approx 0.22$) and $k = \frac{(2-c) (1-2\epsilon)}{2(1-c -2\epsilon)} \approx \frac{0.89 (1-2\epsilon)}{2(0.39 - \epsilon)}$. 
\end{lemma}

\begin{proof}
As seen in the text, the Ky Fan norm of the Bernoulli distribution
\be
||B||_{2^{r+1}-1}= \sum_{i=0}^{m} {2r \choose i} p_+^{2r-i}p_-^{i}
\ee 
can be reformulated using
\be
m \leq \min_{c} \{c r: \sum_{i=0}^{cr} {2r \choose i} \geq 2N -1 \}
\ee 
into finding the minimum $c$ that satisfies the inequality above. 
Note that for $c < 1$
\be
 {2r \choose cr} < \sum_{i=0}^{cr} {2r \choose i} < (cr+1) {2r \choose cr},
\label{c-bound}
\ee 
since ${2r \choose cr}$ is the largest term in the sum. For large $r$ (and consequently large $N = 2^r$), by the Stirling approximation, we have that ${2r \choose cr} \approx 2^{2r H(c/2)}$ where $H(x)$ denotes the binary entropy. Therefore, from $\sum_{i=0}^{cr} {2r \choose i} \geq  2N-1 (\approx 2^r)$ we obtain the condition
\be
2^{2r H(c/2)} = 2^r
\ee
giving the value $c \approx 0.22$ which is asymptotically exact because of the inequalities in (\ref{c-bound}). 

Note that then $||B||_{2^{r+1}-1}$ is trivially lower bounded by ${2r \choose cr} p_{-}^{cr} p_{+}^{(2-c)r}$ as these form a subset of the probabilities appearing in $||B||_{2^{r+1}-1}$. To derive the upper bound, we use the observation that for $0 \leq i \leq cr$,
\ben
\frac{{2r \choose i-1} p_{-}^{i-1} p_{+}^{2r-i+1}}{{2r \choose i} p_{-}^i p_{+}^{2r-i}}  < \frac{i}{2r-i} \frac{p_{+}}{p_{-}} \leq \alpha.
\een
where the constant $\alpha = \frac{c (1+2\epsilon)}{(2-c) (1-2\epsilon)} < 1$ for $\epsilon < 0.39$. 
Iteratively applying the inequality, for $0 \leq i \leq cr$,
\be
{2r \choose i} p_{-}^i p_{+}^{2r-i} < \alpha^{cr-i} {2r \choose cr} p_{-}^{cr} p_{+}^{(2-c)r}.
\ee 
Consequently, we obtain for the Ky Fan norm,
\ben 
||B||_{2^{r+1}-1} &<&  \sum_{i=0}^{cr} \alpha^{cr-i} {2r \choose cr} p_{-}^{cr} p_{+}^{(2-c)r} \nonumber \\
&<& {2r \choose cr} p_{-}^{cr} p_{+}^{(2-c)r} \sum_{i=0}^{\infty} \alpha^i \nonumber \\
&<& \frac{(2-c) (1-2\epsilon)}{2(1-c -2\epsilon)} {2r \choose cr} p_{-}^{cr} p_{+}^{(2-c)r} 
\een
which establishes the upper bound.
\end{proof}

\end{document}